\documentclass[12pt]{article}

\usepackage{hyperref}

\usepackage{amsmath} 
\usepackage{amssymb}
\usepackage{amsfonts}
\usepackage{amsthm}
\usepackage{graphicx}

\newtheorem{theorem}{Theorem}
\newtheorem{lemma}{Lemma}
\newtheorem{corollary}{Corollary}

\theoremstyle{definition}

\begin{document}
	\begin{center}
		\Large
	\textbf{Superoperator master equations for depolarizing dynamics}
	
		\large 
		\textbf{A.E. Teretenkov}\footnote{Department of Mathematical Methods for Quantum Technologies, Steklov Mathematical Institute of Russian Academy of Sciences,
			ul. Gubkina 8, Moscow 119991, Russia\\ E-mail:\href{mailto:taemsu@mail.ru}{taemsu@mail.ru}}
		\end{center}
		
			\footnotesize
			The work is devoted to superoperator master equations. Namely, the  superoperator master equations in the case of the twirling hyperprojector  with respect to the whole unitary group are derived. To be consistent with such a  hyperprojector the free dynamics is assumed to be depolarizing. And it is perturbed by the arbitrary  Gorini--Kossakowski--Sudarshan--Lindblad generator.  The explicit form of the second order master equations are presented in this case.
			\normalsize

\section{Introduction}

A widely used approach to derivation of master equations in non-equilibrium statistical physics is Nakajima-Zwanzig projection formalism \cite{Nakajima1958, Zwanzig1960}. In the theory of open quantum system  it is usually applied to equations for a density matrix. But in  \cite{Teretenkov2022Effective, Teretenkov2024} was developed a modification of such formalism applying  projection methods to  equations for propagators instead of density matrices.It leads to analogs of time-convolutionless master equations \cite{Fulinski1967, Shibata1977, Breuer1999, Breuer2001} which we call superoperator master equations \cite{Teretenkov2024}. We use the term ''hyperprojector'' for projectors which map superoperators to superoperators.  In \cite{Teretenkov2022Effective, Teretenkov2024} the hyperprojector of averaging with respect to free unitary dynamics was considered. In this work we consider another case of hyperprojector. Namely, we consider a hyperprojector of twirling with respect to the whole unitary group. Such a hyperprojector maps any dynamical map to a depolarizing channel, so the superoperator master equations take possibly one of the most simple forms in this case. But this simple case is still non-trivial and some observations in this case can motivate further generalizations. Let us also remark, that the quantum information properties of depolarizing channels are well studied \cite{King2003, Datta2003, Datta2006, Amosov2007}.

In Section \ref{se:superop} we recall the main ideas behind superoperator master equations and represent them in the Schroedinger picture instead of the interaction one. In Section \ref{eq:twirling} we define the twirling hyperpojector and summarize some properties for twirling hyperprojector  with respect to whole unitary group. In Section \ref{se:dynamics} we obtain superoperator master equations  with this specific hyperprojector. Namely, we perturbate a generator of  the semigroup of depolarizing channels with a Gorini--Kossakowski--Sudarshan--Lindblad (GKSL)  generator  \cite{Gorini1976, Lindblad1976} of general form  and obtain explicit form of the second order superoperator master equations in such a case.

In Conclusions we summarize our work and suggest possible directions for further study.  

\section{Superoperator master equations}
\label{se:superop}

In \cite{Teretenkov2024} the superoperator master equations were introduced. Let us recall  the main idea behind them and some results which we need in the present work. The usual projection methods are applied to the equation for the density matrix in the interaction picture
\begin{equation}\label{eq:basicDiffEq}
	\frac{d}{dt} \rho_I(t) = \lambda \mathcal{L}_I(t)\rho_I(t).
\end{equation}
In our work we will consider the case, where $\rho_I(t)$ is a finite-dimensional matrix and $\mathcal{L}_I(t)$ is a continuously differentiable superoperator-valued  function of $t$. In physics applications Equation \eqref{eq:basicDiffEq} is usually a Liouville-von Neumann equation, but the case, where Equation \eqref{eq:basicDiffEq} is a GKSL  equation or its transformation to interaction picture are also actively discussed in recent literature \cite{Saideh2020, Finkelstein-Shapiro2020, Regent2023, Regent2023a, Karasev2023, Meretukov2024}. The solution of Equation \eqref{eq:basicDiffEq} can be represented in the form
\begin{equation*}
	\rho_I(t) = \Phi_I(t, t_0) \rho_I(t_0),
\end{equation*}
where $\Phi_I(t, t_0)$ is a Cauchy matrix or propagator, which is defined as a solution of the Cauchy problem 
\begin{equation}\label{eq:propDef}
	\frac{d}{dt}\Phi_I(t, t_0) = \lambda \mathcal{L}_I(t) \Phi_I(t, t_0). \qquad  \Phi_I(t_0, t_0) = \mathcal{I},
\end{equation}
where $ \mathcal{I}$ is the identity superoperator. The main idea behind superoperator master equations is to apply the projection methods to Equation \eqref{eq:propDef} instead of Equation \eqref{eq:basicDiffEq}. The projectors in such a case map superoperators to superoperators, so we call them hyperprojectors. If one assumes
\begin{equation}\label{eq:identityPreservation}
	\mathfrak{P}( \mathcal{I}) =  \mathcal{I},
\end{equation}
then one obtains the homogeneous linear differential equation for projected dynamics 
\begin{equation}\label{eq:superoperatorMasterEquation}
	\frac{d}{dt}\mathfrak{P}(\Phi_I(t, t_0)) =\mathfrak{K}^I(t) \mathfrak{P}(\Phi_I(t, t_0)),  
\end{equation}
with initial condition $ \mathfrak{P}(\Phi_I(t_0, t_0)) =  \mathcal{I}$. We call Equation \eqref{eq:superoperatorMasterEquation} a superoperator master equation.

$\mathfrak{K}^I(t) $ can be calculated perturbatively using the following theorem \cite[Proposition 1]{Teretenkov2024}.

\begin{theorem}\label{th:cumulantExpansion}
	Asymptotic expansion at fixed $t$ and for $ \mathfrak{K}^I(t) $ as $ \lambda \rightarrow 0 $ has the form
	\begin{equation}\label{eq:cumulantExpansion}
		\mathfrak{K}^I(t) = \sum_{k=1}^{\infty} \lambda^n \mathfrak{K}_k^I(t),
	\end{equation}
	where  $  \mathfrak{K}_n(t) $ are defined as
	\begin{equation}\label{eq:coeffKn}
		\mathfrak{K}_n^I(t) \equiv  \sum_{q=0}^{k-1} (-1)^q \sum_{\sum_{j=0}^q k_j =k, k_j \geqslant 1}  \dot{\mathfrak{M}}_{k_0}^I(t) \mathfrak{M}_{k_1}^I(t)  \ldots \mathfrak{M}_{k_{q}}^I(t),
	\end{equation}
	where the condition $\sum_{j=0}^q k_j =n, k_j \geqslant 1$ means the sum runs over all compositions of the number $n$, and
	\begin{equation}\label{eq:momentsDef}
		\mathfrak{M}_{k}^I(t) \equiv \int_{t_0}^t d t_k  \ldots  \int_{t_0}^{t_2} d t_1  \mathfrak{P}  ( \mathcal{L}_I(t_k) \ldots \mathcal{L}_I(t_1)), \qquad \dot{\mathfrak{M}}_{k}^I(t) \equiv \frac{d}{dt} \mathfrak{M}_{k}^I(t).
	\end{equation}
\end{theorem}
It is analog of the perturbation expansion provided by the Kubo-van Kampen cumulants and similar approaches \cite{Kubo1963, VanKampen1974, VanKampen1974a, Chaturvedi1979, Shibata1980, Nestmann2019}.

In this work we consider the special case of Equation \eqref{eq:propDef} which arises in interaction picture for the equation with the time-independent generator
\begin{equation}\label{eq:propShroedPicture}
	\frac{d}{dt}\Phi(t, t_0) = ( \mathcal{L}_0 + \lambda \mathcal{L}_I ) \Phi(t, t_0), \qquad  \Phi(t_0, t_0)= \mathcal{I}.
\end{equation}
And we assume that the free dynamics commutes with the hyperpojector $ \mathfrak{P}$:
\begin{equation}\label{eq:commAssump}
	\mathfrak{P}(e^{\mathcal{L}_0 t}\Phi) = e^{\mathcal{L}_0 t} \mathfrak{P}(\Phi), \quad  	\mathfrak{P}(\Phi e^{\mathcal{L}_0 t}) =\mathfrak{P}(\Phi)  e^{\mathcal{L}_0 t} 
\end{equation}
for arbitrary superoperator $\Phi $. 

If one defines $ \Phi_I(t, t_0) = e^{- \mathcal{L}_0 (t - t_0)} \Phi(t, t_0) $, then we obtain \eqref{eq:propDef} with $ \mathcal{L}(t) = e^{- \mathcal{L}_0 (t - t_0)} \mathcal{L}_I  e^{\mathcal{L}_0 (t - t_0)} $. From Equation \eqref{eq:superoperatorMasterEquation} under assumption \eqref{eq:commAssump} we have
\begin{equation}\label{eq:projectedpropShroedPicture}
	\frac{d}{dt}\mathfrak{P}(\Phi(t, t_0)) = \mathfrak{K}(t)  \mathfrak{P}(\Phi(t, t_0)),
\end{equation}
where
\begin{equation*}
	\mathfrak{K}(t)  = \mathcal{L}_0 +  e^{\mathcal{L}_0 (t - t_0)} \mathfrak{K}^I(t)  e^{- \mathcal{L}_0 (t - t_0)} .
\end{equation*}

From Theorem \ref{th:cumulantExpansion} we obtain the following asymptotic expansion for $\mathfrak{K}(t)$. 
\begin{corollary}
	\label{cor:cumulantExpansion}
	Let $\Phi(t, t_0)$ be defined as a solution of Cauchy problem \eqref{eq:propShroedPicture} and $\mathfrak{P}$ is a hyperpojector satisfying conditions \eqref{eq:identityPreservation} and \eqref{eq:commAssump}, then $\mathfrak{P}(\Phi(t, t_0))$ satisfies~\eqref{eq:projectedpropShroedPicture} for fixed $t$ and $t_0$ and sufficiently small $\lambda$ and the following asymptotic expansion holds
	\begin{equation*}
		\mathfrak{K}(t) = \sum_{n=0}^{\infty} \lambda^k \mathfrak{K}_k(t), \qquad \lambda \rightarrow 0,
	\end{equation*}
	where
	\begin{equation*}
		\mathfrak{K}_0(t) = \mathcal{L}_0, \qquad 	\mathfrak{K}_n(t) \equiv  \sum_{q=0}^{k-1} (-1)^q \sum_{\sum_{j=0}^q k_j =k, k_j \geqslant 1}  \grave{\mathfrak{M}}_{k_0}(t) \mathfrak{M}_{k_1}(t)  \ldots \mathfrak{M}_{k_{q}}(t),
	\end{equation*}
	with
	\begin{equation}\label{eq:Mk}
		\mathfrak{M}_{k}(t) \equiv  \int_{t_0}^t d t_k  \ldots  \int_{t_0}^{t_2} d t_1  \mathfrak{P}  ( e^{ \mathcal{L}_0(t- t_{k})}  \mathcal{L}_I e^{ \mathcal{L}_0(t_k - t_{k-1})} \ldots  e^{ \mathcal{L}_0(t_2 - t_{1})}\mathcal{L}_I e^{\mathcal{L}_0(t_1-t)})
	\end{equation}
	and
	\begin{align}
		\grave{\mathfrak{M}}_{k}(t) &\equiv  \int_{t_0}^t d t_{k-1}  \ldots  \int_{t_0}^{t_2} d t_1  \mathfrak{P}  (   \mathcal{L}_I e^{ \mathcal{L}_0(t - t_{k-1})} \ldots  e^{ \mathcal{L}_0(t_2 - t_{1})}\mathcal{L}_I e^{\mathcal{L}_0(t_1-t)}) \nonumber\\
		&=\frac{d}{dt}\mathfrak{M}_k(t) - [\mathcal{L}_0,\mathfrak{M}_k(t)]. \label{eq:Mgrk}
	\end{align}
\end{corollary}

In particular, up to the second order we have
\begin{align}
	\mathfrak{K}(t)  =& \mathcal{L}_0 + \lambda \mathfrak{P}(  \mathcal{L}_I  ) \nonumber\\
	&+ \lambda^2  \int_{t_0}^t d t_{1} \left( \mathfrak{P}  (   \mathcal{L}_I   e^{ \mathcal{L}_0(t - t_1)}\mathcal{L}_I e^{-\mathcal{L}_0(t-t_1)}) - \mathfrak{P}  (   \mathcal{L}_I ) \mathfrak{P}  (  e^{ \mathcal{L}_0(t - t_1)}\mathcal{L}_I e^{-\mathcal{L}_0(t-t_1)})\right) +O(\lambda^3). \label{eq:secOrd}
\end{align}

\section{Twirling hyperpojector}
\label{eq:twirling}

In \cite{Teretenkov2022Effective, Teretenkov2024} the hyperpojector of averaging with respect to free unitary dynamics was considered:
\begin{equation}\label{eq:freeAverHyper}
	\mathfrak{P}_{\rm free}(\Phi) = \lim\limits_{T \rightarrow \infty} \frac{1}{T}\int_0^T e^{- i H_0 t} \Phi(e^{i H_0 t} \; \cdot \; e^{- i H_0 t} ) e^{i H_0 t}  dt,
\end{equation}
where $H_0$ is a free Hamiltonian. Superoperator master equations with such a hyperpojector are closely related to effective Hamiltonian theory \cite{Trubilko2019, Trubilko2020, Basharov2021, Teretenkov2022Effective, Teretenkov2022, Aleksashin2023a, Aleksashin2023b}. Here and below we use dot in a superoperator as a place, where the operators should be substituted, when the superoperator is acting on them.

A natural generalization of \eqref{eq:freeAverHyper} is the twirling hyperpojector
\begin{equation*}
	\mathfrak{P}_G(\Phi) = \int_G U^{\dagger} \Phi(U \; \cdot \; U^{\dagger})U dU
\end{equation*}
with respect to a subgroup $G$ of the unitary group $U(n)$ and $dU$ is the Haar measure normalized in such a way that
\begin{equation*}
	\int_G dU = 1.
\end{equation*}
Condition \eqref{eq:identityPreservation} is satisfied automatically for $\mathfrak{P}_G$. 

A physical interpretation of such a hyperprojector can be understood similarly to \cite[Section 4.1]{Teretenkov2024} in terms of two-time measurements. Under Markovian assumption two-time corrleation functions are defined by the quantum regression formula \cite[Section 3.2.4]{Breuer2002}
\begin{equation}\label{eq:twoTimeCorrFun}
	\langle  Y(t) X(t_0) \rangle = \operatorname{Tr}  Y \Phi(t, t_0) (X \rho(t_0)), \qquad t \geqslant t_0.
\end{equation}
Now let us assume the measurement  basis is actually random but fixed so both measurements are performed in the same random basis. And the choice of such a basis is uniquely defined by transformation $U \in G$ from a fixed initial basis. Then correlation function of  transformed operators $X_U = U X U^{\dagger} $ and $Y_U = U Y U^{\dagger} $ in a transformed state $\rho_U(t_0) = U\rho(t_0) U^{\dagger}$ is  measured
\begin{equation*}
	\langle  	Y_U(t) X_U(t_0) \rangle_U =  \operatorname{Tr}  (U Y U^{\dagger} \Phi(t, t_0) ( U X U^{\dagger}  U\rho(t_0) U^{\dagger})) =   \operatorname{Tr}  ( Y U^{\dagger} \Phi(t, t_0) ( U X   \rho(t_0) U^{\dagger}) U ) .
\end{equation*}
If only averaged values with respect to random choice of $U$ are experimentally accessible, then 
\begin{equation*}
	\int_G\langle  	Y_U(t) X_U(t_0) \rangle_U  dU = \int_G \operatorname{Tr}  ( Y U^{\dagger} \Phi(t, t_0) ( U X   \rho(t_0) U^{\dagger}) U ) dU = \operatorname{Tr}  (Y \mathfrak{P}(\Phi(t, t_0)) ( X \rho(t_0))
\end{equation*}
is actually measured. In concrete physical setups similarly to \cite[Section 4.2]{Teretenkov2024} such an averaging can occur effectively due to homogenization on certain scales of time and systems degrees of freedom  without postulating it. Nevertheless, remark that the regression formulae after the averaging for higher order correlation functions can be non-Markovian in this case, similarly to \cite{Teretenkov2023}.

$\mathfrak{P}_{\rm free}$ can be considered as a special $	\mathfrak{P}_G$, where $G$ is a (possibly reducible) representation of $U(1)$. In a certain sense an opposite case is the twirling with respect to  $G = U(n)$. 
\begin{equation}\label{eq:fullAverProj}
	\mathfrak{P}(\Phi) = \int_{U(n)} U^{\dagger} \Phi(U \cdot U^{\dagger})U dU.
\end{equation}
In this work we consider only this hyperprojector here and below, so we omit index $G$. It can be written in the following form  \cite[Corollary 3.11]{Zhang2014}.
\begin{equation}\label{eq:fullAverProjExplicit}
	\mathfrak{P}(\Phi) = \frac{n \operatorname{Tr} (\Phi(I)) - \operatorname{tr}(\Phi)}{n (n^2-1)} (\operatorname{Tr}  X) I +  \frac{n \operatorname{tr}(\Phi) -  \operatorname{Tr} (\Phi(I))}{n (n^2-1)} X,
\end{equation}
where
\begin{equation}\label{eq:traceOfSuperoperator}
	\operatorname{tr} \Phi \equiv \sum_{k, m} \langle k| \Phi(|k\rangle \langle m|) |m\rangle .
\end{equation}
Here and below we assume that the Hilbert space dimension $n \geqslant 2$.

If $\Phi$ is a channel \cite[Equation (3.43)]{Zhang2014}, then $\mathfrak{P}(\Phi)$  is a depolarizing channel
\begin{equation}\label{eq:PofChannel}
	\mathfrak{P}(\Phi) =  \Lambda_p,
\end{equation}
where 
\begin{equation}\label{eq:LambdaDef}
	\Lambda_p X \equiv p X + (1-p) \frac{I}{n} \operatorname{Tr}  X 
\end{equation}
for any matrix $X \in \mathbb{C}^{n \times n}$, and
\begin{equation}\label{eq:pForPhi}
	p = \frac{\operatorname{tr}(\Phi) - 1}{n^2 - 1}.
\end{equation}
Let us remark that formulae \eqref{eq:PofChannel} and \eqref{eq:LambdaDef} are valid for any trace-preserving superoperator, but $\Lambda_p $  is not necessary a channel, if $\Phi$ is not a channel. In terms of $p$  $\Lambda_p $ is a channel if and only if $p \in \left[-\frac{1}{n^2 -1}, 1\right]$ (see, e.g. \cite[Equation (2)]{King2003}). In the case, when $\Phi$ is a channel, $\operatorname{tr}(\Phi)$ can be also written as \cite[Example 46]{Mele2023} 
\begin{equation*}
	\operatorname{tr}(\Phi) = n^2 F_{e}(\Phi),
\end{equation*}
where $ F_{e}$ is the entanglement fidelity of channel $\Phi$
\begin{equation*}
	F_{e}(\Phi) \equiv \langle \Omega | \Phi \otimes I(| \Omega \rangle  \langle \Omega |)| \Omega \rangle,
\end{equation*}
where $| \Omega \rangle$ is the maximally entangled state $ | \Omega \rangle = \frac{1}{\sqrt{n}} \sum_{k=1}^n |k \rangle \otimes |k \rangle$.
\begin{lemma}
	Superoperator \eqref{eq:LambdaDef} and hyperprojector \eqref{eq:fullAverProj} satisfy the following properties
	\begin{enumerate}
		\item (Stationary superoperator)
		\begin{equation}\label{eq:sattionarySuperoperator}
			\mathfrak{P}(\Lambda_p) =  \Lambda_p
		\end{equation}
		\item (Abelian semigroup property) For $p, q \in \mathbb{C}$
		\begin{equation}\label{eq:LambdaComm}
			\Lambda_p \Lambda_q = \Lambda_q \Lambda_p = \Lambda_{pq}.
		\end{equation}
		\item (Commutativity of left or right action of $\Lambda_p$) Let $\Phi$ be a trace-preserving superoperator, then
		\begin{equation}\label{eq:leftRightActionComm}
			\mathfrak{P}(\Lambda_p\Phi) = \Lambda_p\mathfrak{P}(\Phi) = \mathfrak{P}(\Phi)  \Lambda_p = \mathfrak{P}(\Phi \Lambda_p)  .
		\end{equation}
		\item  (Absorbing property) Let $\Phi$ and $\Psi$ be  trace-preserving superoperators, then
		\begin{equation}\label{eq:absorbing}
			\mathfrak{P}(\mathfrak{P}(\Phi)\Psi) = 	\mathfrak{P}(\Phi\mathfrak{P}(\Psi))= \mathfrak{P}(\Phi) \mathfrak{P}(\Psi) 
		\end{equation}
		\item (Commutativity of projected superoperators) Let $\Phi$ and $\Psi$ be  trace-preserving superoperators, then
		\begin{equation}\label{eq:commP}
			\mathfrak{P}(\Phi) \mathfrak{P}(\Psi)= \mathfrak{P}(\Psi) \mathfrak{P}(\Phi) 
		\end{equation}
	\end{enumerate}
\end{lemma}

\begin{proof}
	\begin{enumerate}
		\item 	By definitions \eqref{eq:traceOfSuperoperator} and \eqref{eq:LambdaDef} we have
		\begin{align*}
			\operatorname{tr}(\Lambda_p) = \sum_{k, m} \langle k| \Lambda_p(|k\rangle \langle m|) |m\rangle  = p \sum_{k, m} \langle k|k\rangle \langle m| m\rangle  + (1-p)\sum_{k, m} \langle k| \frac{I}{n} \operatorname{Tr}  (|k\rangle \langle m|)|m\rangle\\
			=p \sum_{k, m} 1 + (1-p) \sum_{k} \langle k| \frac{I}{n} |k\rangle =  p n^2 + (1-p)
		\end{align*}
		So 
		\begin{equation*}
			\frac{\operatorname{tr}(\Lambda_p) - 1}{n^2 - 1} = p
		\end{equation*}
		and by formulae \eqref{eq:PofChannel} and \eqref{eq:LambdaDef} we obtain \eqref{eq:sattionarySuperoperator}.
		\item By definition \eqref{eq:LambdaDef} we have
		\begin{align*}
			\Lambda_p \Lambda_q &= \left( p \mathcal{I} + (1-p) \frac{I}{n} \operatorname{Tr}  (\; \cdot \;)\right)\left( q \mathcal{I} + (1-q) \frac{I}{n} \operatorname{Tr}  (\; \cdot \;)\right)\\
			&= p q \mathcal{I} +  (1-p) q \frac{I}{n} \operatorname{Tr}  (\; \cdot \;) + p (1-q)  \frac{I}{n} \operatorname{Tr}  (\; \cdot \;) + (1-p)  (1-q) \frac{I}{n} \frac{\operatorname{Tr}  I}{n}  \operatorname{Tr}  (\; \cdot \;)\\
			&=  p q \mathcal{I} +  ((1-p) q + + p (1-q) + (1-p)  (1-q) )\frac{I}{n} \operatorname{Tr}  (\; \cdot \;) =  p q \mathcal{I} +  (1 - pq )\frac{I}{n} \operatorname{Tr}  (\; \cdot \;).
		\end{align*}
		\item By definition \eqref{eq:LambdaDef} we have
		\begin{equation}\label{eq:PLambdapPhi}
			\mathfrak{P}(\Lambda_p\Phi) = p \mathfrak{P}(\Phi) +  (1-p) \mathfrak{P}\left( \frac{I}{n} \operatorname{Tr} (\; \cdot \;) \Phi\right)
		\end{equation}
		and
		\begin{equation}\label{eq:LambdapPPhi}
			\Lambda_p\mathfrak{P}(\Phi) = p \mathfrak{P}(\Phi) +  (1-p) \frac{I}{n} \operatorname{Tr} (\; \cdot \;) \mathfrak{P}\left( \Phi\right)
		\end{equation}
		Let as calculate
		\begin{equation}\label{eq:trITrPgiI}
			\operatorname{Tr} ((I \operatorname{Tr} (\; \cdot \;)\Phi)(I)) = \operatorname{Tr} ( I \operatorname{Tr}  \Phi(I)) = n \operatorname{Tr}  \Phi(I)
		\end{equation}
		and
		\begin{align}
			\operatorname{tr}(I \operatorname{Tr} (\; \cdot \;)\Phi) =  \sum_{k, m} \langle k| I \operatorname{Tr} (\Phi(|k\rangle \langle m|)) |m\rangle  &= \sum_{k, m}  \operatorname{Tr} (\Phi(|k\rangle \langle m|) )\langle k| m\rangle \nonumber\\
			&= \sum_k  \operatorname{Tr} (\Phi(|k\rangle \langle k|) ) = \operatorname{Tr}  \Phi(I) \label{eq:trITrPgi}
		\end{align}
		Substituting \eqref{eq:trITrPgiI} and \eqref{eq:trITrPgi} into \eqref{eq:fullAverProjExplicit} we have
		\begin{align*}
			\mathfrak{P}(I \operatorname{Tr} (\; \cdot \;)\Phi) X =& \frac{n  \operatorname{Tr} ((I \operatorname{Tr} (\; \cdot \;)\Phi)(I))  - \operatorname{tr}(I \operatorname{Tr} (\; \cdot \;)\Phi)}{n (n^2-1)} (\operatorname{Tr}  X) I \\
			&+  \frac{n \operatorname{tr}(I \operatorname{Tr} (\; \cdot \;)\Phi) -   \operatorname{Tr} ((I \operatorname{Tr} (\; \cdot \;)\Phi)(I)) }{n (n^2-1)} X\\
			=& \frac{n^2  \operatorname{Tr}  \Phi(I) -  \operatorname{Tr}  \Phi(I)}{n (n^2-1)}  (\operatorname{Tr}  X) I + \frac{n \operatorname{Tr}  \Phi(I) - n \operatorname{Tr}  \Phi(I)}{n (n^2-1)} X\\
			=& \frac{I \operatorname{Tr}  \Phi(I)}{n }  (\operatorname{Tr}  X)  = \frac{I \operatorname{Tr}  \Phi(I)}{n } \operatorname{Tr} (\; \cdot \;)  X
		\end{align*}
		and
		\begin{align*}
			I\operatorname{Tr} (\; \cdot \;) \mathfrak{P}(\Phi) X &=  \frac{n \operatorname{Tr} (\Phi(I)) - \operatorname{tr}(\Phi)}{n (n^2-1)} (\operatorname{Tr}  X)  I\operatorname{Tr} (\; \cdot \;) I +  \frac{n \operatorname{tr}(\Phi) -  \operatorname{Tr} (\Phi(I))}{n (n^2-1)}  I\operatorname{Tr}  X \\
			&=\frac{(n^2-1) \operatorname{Tr} (\Phi(I)) - n \operatorname{tr}(\Phi) + n \operatorname{tr}(\Phi)}{n (n^2-1)} (\operatorname{Tr}  X) I = \frac{I \operatorname{Tr}  \Phi(I)}{n }  (\operatorname{Tr}  X)  
		\end{align*}
		for any matrix $  X \in \mathbb{C}$. Substituting it to \eqref{eq:PLambdapPhi} and \eqref{eq:LambdapPPhi} we obtain
		\begin{equation*}
			\mathfrak{P}(\Lambda_p\Phi) = p \mathfrak{P}(\Phi) +  (1-p)   \frac{I \operatorname{Tr}  \Phi(I)}{n } \operatorname{Tr} (\; \cdot \;)   \mathfrak{P}(\Phi) = \Lambda_p	\mathfrak{P}(\Phi).
		\end{equation*}
		Similarly, we have
		\begin{equation*}
			\mathfrak{P}(\Phi\Lambda_p) = \mathfrak{P}(\Phi) \Lambda_p.
		\end{equation*}
		As $ \mathfrak{P}(\Phi) = \Lambda_{p'}$, where $p'$ is defined by formula \eqref{eq:leftRightActionComm} $	p' = (n^2 - 1)^{-1} (\operatorname{tr}(\Phi) - 1)$, then due to \eqref{eq:LambdaComm} we have
		\begin{equation*}
			\Lambda_p	\mathfrak{P}(\Phi) = \Lambda_p\Lambda_{p'} = \Lambda_{p'}\Lambda_p = 	\mathfrak{P}(\Phi) \Lambda_p.
		\end{equation*}
		\item Similarly, by Equation \eqref{eq:leftRightActionComm} we have
		\begin{equation}
			\mathfrak{P}(\mathfrak{P}(\Phi)\Psi) = \mathfrak{P}(\Lambda_{p}\Psi)  = \Lambda_{p} \mathfrak{P}(\Psi) = \mathfrak{P}(\Phi) \mathfrak{P}(\Psi) ,
		\end{equation}
		where $	p = (n^2 - 1)^{-1} (\operatorname{tr}(\Phi) - 1)$, and
		\begin{equation}
			\mathfrak{P}(\Phi \mathfrak{P}(\Psi)) = \mathfrak{P}(\Phi \Lambda_{q} )  =\mathfrak{P}(\Phi )  \Lambda_{q}= \mathfrak{P}(\Phi) \mathfrak{P}(\Psi) ,
		\end{equation}
		where $	q = (n^2 - 1)^{-1} (\operatorname{tr}(\Psi) - 1)$.
		\item By formula \eqref{eq:LambdaComm} we have
		\begin{equation*}
			\mathfrak{P}(\Phi) \mathfrak{P}(\Psi)= \Lambda_p \Lambda_q  = \Lambda_q \Lambda_p = \mathfrak{P}(\Psi) \mathfrak{P}(\Phi),
		\end{equation*}
		where $	p = (n^2 - 1)^{-1} (\operatorname{tr}(\Phi) - 1)$ and $	q = (n^2 - 1)^{-1} (\operatorname{tr}(\Psi) - 1)$.
	\end{enumerate}
\end{proof}

In \cite{Teretenkov2024} it was mentioned  that the difference between master equations in usual projection   approach and the one with hyperpojectors can be seen as change $\mathcal{P} \; \cdot \mathcal{P}$ to $\mathfrak{P}$. Here $ \mathcal{P}$ is a usual projector. This difference mostly consists in the fact that absorbing property \eqref{eq:absorbing} obviously holds for $\mathcal{P} \; \cdot \mathcal{P}$, but can be violated for arbitrary hyperpojectors. Thus, absorbing property is crucial for similarity to  usual projection  approach.  This property can be considered as a superoperator analog of associativity of the Choi-Effros product \cite{Choi1977, Yashin2022}.

For any generator $\mathcal{L}$ of trace-prserving semigroup, i.e. a superoperator $ \mathcal{L}$ such that $\operatorname{Tr}  \mathcal{L}(X) = 0$, Equation \eqref{eq:fullAverProjExplicit} is simplified to 
\begin{equation}\label{eq:PL}
	\mathfrak{P}(\mathcal{L})  = \frac{\operatorname{Tr} (\mathcal{L}) }{n^2-1} \left(\mathcal{I} - \frac{I}{n} \operatorname{Tr}  (\; \cdot \;)\right).
\end{equation}

\begin{lemma}
	\begin{enumerate}
		Let $\mathfrak{P}$ be defined by formula \eqref{eq:fullAverProj} and $	\Lambda_{p}$ be defined by \eqref{eq:LambdaDef},  $p\in \mathbb{C}$.
		\item Let $\mathcal{L}$ be a superoperator, such that $\operatorname{Tr}  \mathcal{L}(X) = 0$ for all $X \in \mathbb{C}^{n \times n}$, then
		\begin{equation}\label{eq:LambdaL}
			\Lambda_{p} \mathcal{L} = p  \mathcal{L}
		\end{equation}
		and
		\begin{equation}\label{eq:LLambda}
			\mathcal{L} \Lambda_{p}= p  \mathcal{L} + (1-p)\frac{\mathcal{L}(I)}{n} \operatorname{Tr}  (\; \cdot \;) .
		\end{equation}
		\item  Let $\mathcal{L}$ be a superoperator, such that $\operatorname{Tr}  \mathcal{L}(X) = 0$ for all $X \in \mathbb{C}^{n \times n}$, then
		\begin{equation}\label{eq:PLambdaL}
			\mathfrak{P}(\Lambda_{p} \mathcal{L} ) = \Lambda_{p}  \mathfrak{P}(\mathcal{L} ) =   \mathfrak{P}(\mathcal{L} ) \Lambda_{p} =   \mathfrak{P}(\mathcal{L} \Lambda_{p}) 
		\end{equation}
		\item Let $\mathcal{L}_j$, $j=1, \ldots k$, be superoperators, such that   $\operatorname{Tr}  \mathcal{L}_{j}(X) = 0$ for all $X \in \mathbb{C}^{n \times n}$, then
		\begin{equation}\label{eq:LambdaCommFromP}
			\mathfrak{P}(\Lambda_{p_1} \mathcal{L}_1 \ldots  \Lambda_{p_k} \mathcal{L}_k \Lambda_{p_{k+1}} ) =  \Lambda_{p_1 \ldots p_{k+1}} \mathfrak{P}( \mathcal{L}_1 \ldots   \mathcal{L}_k ).
		\end{equation}
	\end{enumerate}
\end{lemma}
\begin{proof}
	\begin{enumerate}
		\item As $\operatorname{Tr}  \mathcal{L}(\; \cdot \;) = 0$, then using definition \eqref{eq:LambdaDef} we obtain
		\begin{equation*}
			\Lambda_{p} \mathcal{L} = p  \mathcal{L} + (1-p)\frac{I}{n} \operatorname{Tr}  \mathcal{L}(\; \cdot \;) = p \mathcal{L},
		\end{equation*}
		and similarly we obtain \eqref{eq:LLambda}.
		\item If $\operatorname{Tr}  \mathcal{L}(\; \cdot \;) = 0$, then $ e^{\mathcal{L} t} $ is a trace preserving map.	From Equation \eqref{eq:leftRightActionComm} we have
		\begin{equation*}
			\mathfrak{P}(  \Lambda_p e^{\mathcal{L} t})  =  \Lambda_p \mathfrak{P}(e^{\mathcal{L} t})  .
		\end{equation*}
		Expanding in the Taylor series around $t=0$ we have
		\begin{equation*}
			\mathfrak{P}(  \Lambda_p ) + 	\mathfrak{P}(  \Lambda_p \mathcal{L}) t + O(t^2)  =   \Lambda_p  \mathfrak{P}(\mathcal{I} ) + \Lambda_p	\mathfrak{P}(   \mathcal{L}) t + O(t^2), \qquad t \rightarrow 0
		\end{equation*}
		The zeroth order in $t$ is equivalent to \eqref{eq:sattionarySuperoperator}. The first order leads to $	\mathfrak{P}(\Lambda_{p} \mathcal{L} ) = \Lambda_{p}  \mathfrak{P}(\mathcal{L} )$. Similarly, other equalities of \eqref{eq:PLambdaL} can be obtained.
		\item From formulae \eqref{eq:LambdaL} and \eqref{eq:PLambdaL} we have
		\begin{align*}
			\mathfrak{P}(\Lambda_{p_1} \mathcal{L}_1 \ldots  \Lambda_{p_k} \mathcal{L}_k \Lambda_{p_{k+1}} ) = p_1 \ldots p_k \mathfrak{P}( \mathcal{L}_1 \ldots   \mathcal{L}_k \Lambda_{p_{k+1}} )
			&=p_1 \ldots p_k\mathfrak{P}(  \Lambda_{p_{k+1}}  \mathcal{L}_1 \ldots   \mathcal{L}_k ) \\
			= p_1 \ldots p_{k+1} \mathfrak{P}( \mathcal{L}_1 \ldots   \mathcal{L}_k ) =  \mathfrak{P}(\Lambda_{p_1 \ldots p_{k+1}} \mathcal{L}_1 \ldots   \mathcal{L}_k )
			&= \Lambda_{p_1 \ldots p_{k+1}} \mathfrak{P}( \mathcal{L}_1 \ldots   \mathcal{L}_k ).
		\end{align*}
	\end{enumerate}
\end{proof}

Let us remark that from \eqref{eq:LambdaL} and \eqref{eq:LLambda} we have
\begin{equation*}
	[\mathcal{L}, \Lambda_{p}]= p(1-p)\frac{\mathcal{L}(I)}{n} \operatorname{Tr}  (\; \cdot \;) ,
\end{equation*}
so for generators of non-unital semigroups, i.e. $ \mathcal{L}(I) \neq 0$, $\mathcal{L}$ and $\Lambda_{p}$ do not commute. But \eqref{eq:LambdaCommFromP} means that they are effectively commutative inside the action of hyperprojector $\mathfrak{P}$. 

\section{Depolarizing dynamics and depolarization rate}
\label{se:dynamics}

Let us consider the dynamics with the generator 
\begin{equation}\label{eq:generatorOfDepolarizing}
	\mathcal{L}_0 = \gamma (\Lambda_p - \mathcal{I}), \qquad \gamma >0, \qquad p \in \left[-\frac{1}{n^2 -1}, 1\right],
\end{equation}
or more explicitly 
\begin{equation*}
	\mathcal{L}_0 X =  -  (1-p) \gamma \left(X - \frac{I}{n} \operatorname{Tr}  X\right).
\end{equation*}

Let us remark that $\mathcal{L}_0 $ is a GKSL generator \cite[Lemma 1]{Wolf2008}, because using the Kraus representation of
\begin{equation*}
	\Lambda_p = \sum_k W_k \; \cdot \; W_k^{\dagger}, \qquad  \sum_k  W_k^{\dagger} W_k = I,
\end{equation*}
we have 
\begin{equation*}
	\mathcal{L}_0 = \gamma \left(\Lambda_p  - \frac12 \{I, \; \cdot \;\}\right) =  \gamma \sum_k \left(W_k \; \cdot \; W_k^{\dagger}  - \frac12 \{W_k^{\dagger} W_k, \; \cdot \;\}\right).
\end{equation*}
Thus, $ \sqrt{\gamma} W_k$ can be taken as Lindblad (jump) operators. Depolarizing channel $\Lambda_p $ is not specific here, and $\mathcal{L}_0$ has GKSL form for any channels instead of  $\Lambda_p $. Moreover,  in general case it arises as averaging of evolution  $\Lambda_p^k$ with discrete time $k$ with respect to the Poisson process with intensity $\gamma$ \cite[Section 6]{Nosal2022}.

Let us show that free dynamics is a depolarizing channel at any fixed time.

\begin{lemma}\label{lem:freeDynamics}
	Let $\mathcal{L}_0 $ be defined by formula \eqref{eq:generatorOfDepolarizing}, then
	\begin{equation}\label{eq:freeDynamics}
		e^{\mathcal{L}_0 t} = e^{-(1-p) \gamma t} \mathcal{I} + (1-e^{-(1-p) \gamma t})\frac{I}{n} \operatorname{Tr}  (\; \cdot \;) =  \Lambda_{ e^{-(1-p) \gamma t} },
	\end{equation}
	where $\Lambda_p$ is defined by formula \eqref{eq:LambdaDef}.
\end{lemma}

\begin{proof}
	As
	\begin{equation*}
		\left(\frac{I}{n} \operatorname{Tr}  (\; \cdot \;)\right)^2 = \frac{I}{n} \operatorname{Tr}  (\; \cdot \;) \frac{I}{n} \operatorname{Tr}  (\; \cdot \;) = \frac{I}{n} \operatorname{Tr}  (\; \cdot \;) 
	\end{equation*}
	then via the Taylor expansion one has
	\begin{equation*}
		e^{ (1-p) \gamma t \frac{I}{n} \operatorname{Tr}  (\; \cdot \;) }  = (e^{(1-p) \gamma t}-1)\frac{I}{n} \operatorname{Tr}  (\; \cdot \;) .
	\end{equation*}
	Then we have
	\begin{align*}
		e^{\mathcal{L}_0 t} = e^{-(1-p) \gamma t} e^{ (1-p) \gamma t \frac{I}{n} \operatorname{Tr}  (\; \cdot \;) } &=  e^{-(1-p) \gamma t} \left(\mathcal{I} + (e^{(1-p) \gamma t}-1)\frac{I}{n} \operatorname{Tr}  (\; \cdot \;)\right)\\
		& =  e^{-(1-p) \gamma t} \mathcal{I} + (1-e^{-(1-p)  \gamma t})\frac{I}{n} \operatorname{Tr}  (\; \cdot \;).
	\end{align*}
	Thus, we have obtained Equation \eqref{eq:freeDynamics}.
\end{proof}

So it is natural to call $\tilde{\gamma}_0 = (1-p) \gamma$ a free depolarization rate. Now, let us justify condition \eqref{eq:commAssump} for free dynamics.

\begin{lemma}
	Condition \eqref{eq:commAssump} is met for $\mathfrak{P}$  defined by Equation \eqref{eq:fullAverProj} and $\mathcal{L}_0$ defined by Equation \eqref{eq:generatorOfDepolarizing}.
\end{lemma}

\begin{proof}
	From lemma \ref{lem:freeDynamics} and Equation \eqref{eq:leftRightActionComm} we have
	\begin{equation*}
		\mathfrak{P}(e^{\mathcal{L}_0 t} \Phi) = \mathfrak{P}(\Lambda_{ e^{-(1-p) \gamma t} } \Phi) = \Lambda_{ e^{-(1-p) \gamma t} }\mathfrak{P}( \Phi) = e^{\mathcal{L}_0 t} \mathfrak{P}( \Phi)
	\end{equation*}
	and similarly
	\begin{equation*}
		\mathfrak{P}(\Phi e^{\mathcal{L}_0 t}) = \mathfrak{P}(\Phi \Lambda_{ e^{-(1-p) \gamma t} }) = \mathfrak{P}( \Phi) \Lambda_{ e^{-(1-p) \gamma t} } =  \mathfrak{P}( \Phi) e^{\mathcal{L}_0 t}.
	\end{equation*}
\end{proof}

Thus, the projector \eqref{eq:fullAverProj} is consistent with free dynamics $e^{\mathcal{L}_0 t}$, where $\mathcal{L}_0$ is defined by Equation \eqref{eq:generatorOfDepolarizing}. Now let us consider perturbated dynamics \eqref{eq:propShroedPicture} with arbitrary superoperator $\mathcal{L}_I$ such that $\mathcal{L}_I(X) = 0$ for any matrix $X \in \mathbb{C}^{n \times n}$.

\begin{theorem}
	Let $\Phi(t, t_0)$ be defined as a solution of Cauchy problem \eqref{eq:propShroedPicture}, where $\mathcal{L}_0$ is defined by \eqref{eq:generatorOfDepolarizing} and  $\operatorname{Tr}  \mathcal{L}_I(X) = 0$ for any matrix $X \in \mathbb{C}^{n \times n}$. Let $\mathfrak{P}$ be defined by \eqref{eq:fullAverProj}. Then $\mathfrak{P}(\Phi(t, t_0))$ satisfies~\eqref{eq:projectedpropShroedPicture} for fixed $t$ and $t_0$ and sufficiently small $\lambda$ and the following asymptotic expansion holds
	\begin{equation*}
		\mathfrak{K}(t) = \sum_{n=0}^{\infty} \lambda^k \mathfrak{K}_k(t), \qquad \lambda \rightarrow 0,
	\end{equation*}
	where $\mathfrak{K}_0(t) = \mathcal{L}_0 $ and
	\begin{equation}\label{eq:cumulantsForProjExpl}
		\mathfrak{K}_{k}(t) =  (t-t_0)^{k-1} \sum_{q=0}^{k-1} (-1)^q \sum_{\sum_{j=0}^q k_j =k, k_j \geqslant 1} \frac{1}{(k_0 - 1)!  k_1! \cdots k_q!} \mathfrak{P}  (\mathcal{L}_I^{k_0-1} ) \ldots (\mathcal{L}_I^{k_1} ) .
	\end{equation}
\end{theorem}

\begin{proof}
	From \eqref{eq:LambdaCommFromP} and \eqref{eq:freeDynamics} 
	\begin{align*}
		&\mathfrak{P}  ( e^{ \mathcal{L}_0(t- t_{k})}  \mathcal{L}_I e^{ \mathcal{L}_0(t_k - t_{k-1})} \ldots  e^{ \mathcal{L}_0(t_2 - t_{1})}\mathcal{L}_I e^{\mathcal{L}_0(t_1-t)}) \\
		&= \mathfrak{P}  ( \Lambda_{ e^{-(1-p) \gamma (t - t_k)} } \mathcal{L}_I \Lambda_{ e^{-(1-p) \gamma (t_k - t_{k-1})} } \ldots  \Lambda_{ e^{-(1-p) \gamma (t_2 - t_1)} }\mathcal{L}_I \Lambda_{ e^{-(1-p) \gamma (t_1 - t)} }) \\
		&= \Lambda_{1} \mathfrak{P}  (\mathcal{L}_I^k ) = \mathfrak{P}  (\mathcal{L}_I^k ).
	\end{align*}
	Substituting it into \eqref{eq:Mk} we have
	\begin{equation*}
		\mathfrak{M}_{k}(t) =  \mathfrak{P}  (\mathcal{L}_I^k )\frac{(t-t_0)^k}{k!}.
	\end{equation*}
	Similarly, substituting it into \eqref{eq:Mgrk} we have
	\begin{equation*}
		\grave{\mathfrak{M}}_{k}(t) =  \mathfrak{P}  (\mathcal{L}_I^k )\frac{(t-t_0)^{k-1}}{(k-1)!}.
	\end{equation*}
	Thus, we have
	\begin{align*}
		\grave{\mathfrak{M}}_{k_0}(t) \mathfrak{M}_{k_1}(t)  \ldots \mathfrak{M}_{k_{q}}(t) =   \mathfrak{P}  (\mathcal{L}_I^{k_0} )\frac{(t-t_0)^{k_0-1}}{(k_0-1)!} \mathfrak{P}  (\mathcal{L}_I^{k_1} )\frac{(t-t_0)^{k_1}}{k_1!} \ldots \mathfrak{P}  (\mathcal{L}_I^{k_q} )\frac{(t-t_0)^{k_q}}{k_q!}\\
		=\frac{1}{(k_0 - 1)!  k_1! \cdots k_q!} \mathfrak{P}  (\mathcal{L}_I^{k_0-1} ) \ldots (\mathcal{L}_I^{k_1} ) (t-t_0)^{\sum_{j=0}^q k_j-1}
	\end{align*}
	and by Corollary \ref{cor:cumulantExpansion} we obtain \eqref{eq:cumulantsForProjExpl}.
\end{proof}

In particular, second order generator \eqref{eq:secOrd} has the form 
\begin{equation}\label{eq:secOrdTwirlingAllUnitary}
	\mathfrak{K}(t)  = \mathcal{L}_0 + \lambda \mathfrak{P}(  \mathcal{L}_I  ) + \lambda^2  (t-t_0) \left( \mathfrak{P}  (   \mathcal{L}_I^2) - (\mathfrak{P}  (   \mathcal{L}_I ))^2 \right) +O(\lambda^3), \qquad \lambda \rightarrow 0.
\end{equation}

To calculate it more explicitly we need the following lemma.
\begin{lemma}\label{lemm:trace}
	Let $A, B \in \mathbb{C}^{n \times n}$, then
	\begin{equation*}
		\operatorname{tr}( A \; \cdot \; B) =\operatorname{Tr}  A \operatorname{Tr}  B,
	\end{equation*}
	in particular, 
	\begin{equation*}
		\operatorname{tr}( A \; \cdot \;  ) = 	\operatorname{tr}(  \; \cdot \;  A ) = n \operatorname{Tr}  A , \qquad \operatorname{tr}( B \; \cdot \; B^{\dagger}) = |\operatorname{Tr}  B|^2.
	\end{equation*}
\end{lemma}

\begin{proof}
	Using definition \eqref{eq:traceOfSuperoperator}, we have  \cite[Equation (7)]{Ende2023}
	\begin{equation*}
		\operatorname{tr}( A \; \cdot \; B) \equiv \sum_{k, m} \langle k|A|k\rangle \langle m| B|m\rangle  = \operatorname{Tr}  A \operatorname{Tr}  B.
	\end{equation*}
	In particular, $\operatorname{tr}( A \; \cdot \;  ) = \operatorname{tr}( A \; \cdot \;  I) = \operatorname{Tr}  A \operatorname{Tr}  I = n\operatorname{Tr}  A$, similarly, $\operatorname{tr}(  \; \cdot \;  A) = n\operatorname{Tr}  A $, and $\operatorname{tr}( B \; \cdot \; B^{\dagger}) = \operatorname{Tr}  B \operatorname{Tr}  B^{\dagger}  = |\operatorname{Tr}  B|^2$.
\end{proof}

Now let us assume that $\mathcal{L}_I$ has a GKSL form
\begin{equation}\label{eq:GKSLform}
	\mathcal{L}_I = - i [H, \; \cdot \;] + \sum_j\left(L_j \; \cdot  \; L_j^{\dagger}  -  \frac12 \{L_j^{\dagger} L_j, \; \cdot  \;\}  \right), \qquad H = H^{\dagger}.
\end{equation}
The GKSL form is invariant under transformation \cite[Equation (3.73)]{Breuer2002}
\begin{equation*}
	L_j \rightarrow L_j + c_j I, \qquad H \rightarrow H + \frac{1}{2i} \sum_j (c_j^* L_j - c_j L_j^{\dagger}), \qquad c_i \in \mathbb{C}.
\end{equation*}
So one can choose $L_j$ in such a way that $\operatorname{Tr}  L_j = 0$. 

\begin{lemma}\label{lemm:GKSLtraces}
	If $\mathcal{L}_I $ has GSKL form \eqref{eq:GKSLform} and $\operatorname{Tr}  L_j = 0$, then 
	\begin{equation*}
		\frac{1}{n^2}\operatorname{tr} \mathcal{L}_I  = -  \langle  G\rangle
	\end{equation*}
	and
	\begin{equation*}
		\frac{1}{n^2}\operatorname{tr} \mathcal{L}_I^2 = -2 ( \langle H^2 \rangle -\langle H \rangle^2) + \sum_{j,k} |\langle L_j L_k \rangle|^2 + \frac12 \langle G^2 \rangle + \frac12 \langle G \rangle^2,
	\end{equation*}
	where
	\begin{equation*}
		G = \sum_j L_j^{\dagger} L_j
	\end{equation*}
	and $\langle \; \cdot \; \rangle$ is an average with respect to chaotic state, i.e.  $\langle \; \cdot \; \rangle = n^{-1} \operatorname{Tr} ( \; \cdot \; )$.
\end{lemma}

\begin{proof} 
	Applying  \ref{lemm:trace} to \eqref{eq:GKSLform} we have  
	\begin{align*}
		\operatorname{tr} \mathcal{L}_I = - i n \operatorname{Tr}  H + i n \operatorname{Tr}  H + \sum_j (\operatorname{Tr}  L_j  \operatorname{Tr}  L_j^{\dagger} - n \operatorname{Tr}  (L_j^{\dagger} L_j)) 
		= - n^2 \sum_j \langle  L_j^{\dagger} L_j\rangle = - n^2  \langle  G\rangle .
	\end{align*}
	Taking the square of GKSL generator \eqref{eq:GKSLform}
	\begin{align*}
		&\mathcal{L}_I^2 =-  [H, [H, \; \cdot \;] ]  \\
		&- i   \sum_j\left(L_j [H, \; \cdot \;] L_j^{\dagger}  -  \frac12 \{L_j^{\dagger} L_j,  [H, \; \cdot \;]\}\right)
		-i\sum_j\left([H, L_j \; \cdot  \; L_j^{\dagger} ] -  \frac12 [H, \{L_j^{\dagger} L_j, \; \cdot  \;\} ] \right) \\
		&+\sum_{j,k}\left(L_j L_k \; \cdot  \; L_k^{\dagger} L_j^{\dagger}  -  \frac12 \{L_j^{\dagger} L_j, L_k \; \cdot  \; L_k^{\dagger}\}  \right)
		-\frac12 \sum_{j,k}\left(L_j  \{L_k^{\dagger} L_k, \; \cdot  \;\}  L_j^{\dagger}  -  \frac12 \{L_j^{\dagger} L_j,  \{L_k^{\dagger} L_k, \; \cdot  \;\}\}  \right).
	\end{align*}
	Then using Lemma \ref{lemm:trace} we obtain
	\begin{align*}
		\operatorname{tr} \mathcal{L}_I^2 =&-2( n  \operatorname{Tr}  H^2- (\operatorname{Tr}  H)^2 ) 
		- i \sum_j\left( 2\operatorname{Tr}  L_j H\operatorname{Tr}  L_j^{\dagger} - 2 \operatorname{Tr}  L_j \operatorname{Tr}  L_j^{\dagger} H\right)\\
		&+\sum_{j,k}\left(\operatorname{Tr}  L_j L_k \operatorname{Tr}  L_j^{\dagger} L_k^{\dagger}  -  \operatorname{Tr}  L_k \operatorname{Tr}   L_k^{\dagger}  L_j^{\dagger} L_j     - \operatorname{Tr}  L_j^{\dagger} L_j L_k \operatorname{Tr}   L_k^{\dagger} \right)\\
		&+\frac12 \sum_{j,k} (\operatorname{Tr}  L_j^{\dagger} L_j \operatorname{Tr}  L_k^{\dagger} L_k + n  \operatorname{Tr}  L_j^{\dagger} L_j   L_k^{\dagger} L_k )\\
		=&  -2 n^2( \langle H^2 \rangle -\langle H \rangle^2) + n^2\sum_{j,k} |\langle L_j L_k \rangle|^2 + \frac12 n^2\bigl\langle ( \sum_j L_j^{\dagger} L_j )^2 \bigr\rangle + \frac12 n^2 \bigl\langle  \sum_j L_j^{\dagger} L_j  \bigr\rangle^2\\
		=& n^2 \left(-2 ( \langle H^2 \rangle -\langle H \rangle^2) + \sum_{j,k} |\langle L_j L_k \rangle|^2 + \frac12 \langle G^2 \rangle + \frac12  \langle  G \rangle^2\right).
	\end{align*}
	
\end{proof}

\begin{theorem}
	If $\mathcal{L}_I $ has GSKL form \eqref{eq:GKSLform} and $\operatorname{Tr}  L_j = 0$, then the second order generator  \eqref{eq:secOrdTwirlingAllUnitary} takes the form
	\begin{equation}\label{eq:secOrder}
		\mathfrak{K}(t)  =  -\tilde{\gamma}_{\lambda}(t)  \left(\mathcal{I} - \frac{I}{n} \operatorname{Tr}  (\; \cdot \;)\right)+ O(\lambda^3), \qquad \lambda \rightarrow 0,
	\end{equation}
	where
	\begin{align}
		&\tilde{\gamma}_{\lambda}(t) = \gamma(1-p) + \lambda \frac{n^2 }{n^2-1} \langle  G\rangle \nonumber\\
		& + \lambda^2(t-t_0) \left(\frac{n^2}{n^2-1}\biggl(2 ( \langle H^2 \rangle -\langle H \rangle^2) - \sum_{j,k} |\langle L_j L_k \rangle|^2 - \frac12 \langle G^2 \rangle - \frac12 \langle G \rangle^2\biggr) + \frac{n^4}{(n^2-1)^2} \langle  G\rangle^2\right) .\label{eq:plambdat}
	\end{align}
\end{theorem}

\begin{proof}
	Using Equation \eqref{eq:PL} and Lemma \ref{lemm:GKSLtraces}, we have
	\begin{equation*}
		\mathfrak{P}(\mathcal{L}_I)  = \frac{\operatorname{tr}(\mathcal{L}_I) }{n^2-1} \left(\mathcal{I} - \frac{I}{n} \operatorname{Tr}  (\; \cdot \;)\right) =   - \frac{n^2 }{n^2-1} \langle  G\rangle \left(\mathcal{I} - \frac{I}{n} \operatorname{Tr}  (\; \cdot \;)\right)
	\end{equation*}
	and
	\begin{align*}
		&\mathfrak{P}(\mathcal{L}_I^2) - (\mathfrak{P}(\mathcal{L}_I))^2 =  \frac{(n^2-1)\operatorname{tr}(\mathcal{L}_I^2) -(\operatorname{tr}\mathcal{L}_I)^2}{(n^2-1)^2} \left(\mathcal{I} - \frac{I}{n} \operatorname{Tr}  (\; \cdot \;)\right)\\
		&=	\left(\frac{n^2}{n^2-1}\biggl(-2 ( \langle H^2 \rangle -\langle H \rangle^2) + \sum_{j,k} |\langle L_j L_k \rangle|^2 + \frac12 \langle G^2 \rangle + \frac12 \langle G \rangle^2\biggr)- \frac{n^4}{(n^2-1)^2} \langle  G\rangle^2\right) \times\\
		&\qquad \times \left(\mathcal{I} - \frac{I}{n} \operatorname{Tr}  (\; \cdot \;)\right).
	\end{align*}
	Substituting it into \eqref{eq:secOrdTwirlingAllUnitary} we obtain \eqref{eq:secOrder} with $\tilde{\gamma}_{\lambda}(t)$ defined by \eqref{eq:plambdat}.
\end{proof}

Remark that the first non-zero correction   to the leading order in \eqref{eq:plambdat} both in the dissipative case ($L_i \neq 0$) and in the non-trivial unitary case ($L_i \neq 0$, $H \neq 0$) is always positive. Thus, $\tilde{\gamma}_{\lambda}(t)$, which can be called  depolarization rate, increases due to GKSL perturbations to the free depolarizing dynamics.

Let us consider a general GKSL generator for a two-level system in  the weak coupling limit at equilibrium \cite[p. 83]{Alicki2007} as an example to illustrate Equation \eqref{eq:plambdat}:
\begin{align*}
	\mathcal{L}_I(\rho) =&- i[\omega_0 \sigma_+ \sigma_-, \rho]\\
	& +\gamma_0 (N+1) \left( \sigma_- \rho \sigma_+ - \frac12 \sigma_+ \sigma_- \rho - \frac12 \rho \sigma_+ \sigma_- \right)  \\
	&+\gamma_0 N \left( \sigma_+ \rho \sigma_- - \frac12 \sigma_- \sigma_+ \rho - \frac12 \rho \sigma_- \sigma_+ \right)\\
	&+\gamma_{\rm ph} (\sigma_z \rho \sigma_z - \rho),
\end{align*}
where $\gamma_0 \geqslant 0$, $\gamma_{\rm ph} \geqslant 0$, $N \geqslant 0$ and $\omega_0 \in \mathbb{R}$, then Equation \eqref{eq:plambdat} for depolarization rate takes the form
\begin{equation*}
	\tilde{\gamma}_{\lambda}(t)  = (1-p) \gamma + \lambda \frac{4}{3 } \left(\gamma_0 \left(N +  \frac{1}{2}\right) + \gamma_{\rm ph}\right) + \lambda^2 \frac{4 (2   \omega_0^2 - \gamma_{\rm ph}^2) - \gamma_0^2(2N+1)^2}{6}(t-t_0).
\end{equation*}
It  describes explicitly  the influence of thermal decay  (described by $\gamma_0$ and $N$), pure dephasing  (described by $\gamma_{\rm ph}$) and coherent oscillations (described by $\omega_0$) on the depolarization process.

\section{Conclusions}

We have derived the superoperator master equations for the twirling hyperprojector with respect to the whole unitary group. They describe  depolarizing dynamics, which is natural as the depolarizing channels are stationary points of such a hyperprojector. By our approach we describe the influence of GKSL perturbations on the depolarization process, calculating their contributions into depolarization rate.

From the point of view of general development of superoperator master equation methods it is interesting to mention the absorbing property of such a projector. It makes the superoperator master equation approach much closer to the usual one. So it is worth studying the class of hyperprojectors with the absorbing property. The fact that non-commutative free dynamics and interaction generators behave like commutative ones under the hyperprojector also seems to be an interesting property. So possibly some conditions for it can be found in a more general setup. 

And of course our work makes a step to  the superoperator master equations with twirling hyperprojectors in the case of other groups. They can be used to analyze the establishing rates for  the  symmetries correspondent to such groups in  open quantum systems.

\section*{Acknowledgments}

The author thanks F.~vom~Ende, A.\,Yu.~Karasev, E.\,A.~Kuryanovich and V.\,I.~Yashin for the fruitful discussion of the problems considered in the work.

\end{document}